\documentclass[preprint,12pt,3p]{elsarticle}

\usepackage{amssymb}
\usepackage{graphicx}
\usepackage{amsmath}

\usepackage{tikz}
\usetikzlibrary{arrows}
\usetikzlibrary{shapes}
\usetikzlibrary{decorations.markings}
\usepackage{caption}

\journal{Journal of Latex templates}

\newtheorem{theorem}{Theorem}
\newtheorem{proposition}{Proposition}
\newtheorem{conjecture}{Conjecture}
\newtheorem{lemma}{Lemma}
\newtheorem{corollary}{Corollary}
\newtheorem{problem}{Problem}

\newenvironment{proof}{ {\bf Proof:}} {$\Box$}

\begin{document}

\begin{frontmatter}

\title{$S_{12}$ and $P_{12}$-colorings of cubic graphs}

\author[label0]{Anush Hakobyan}
\address[label0]{Department of Informatics and Applied Mathematics, Yerevan State University, Yerevan, 0025, Armenia}
\ead{ashunik94@gmail.com}






\author[label1]{Vahan Mkrtchyan\corref{cor1}}
\address[label1]{Dipartimento di Informatica,
Universita degli Studi di Verona, Strada le Grazie 15, 37134 Verona, Italy}
\ead{vahanmkrtchyan2002@ysu.am}

\cortext[cor1]{Corresponding author}








\begin{abstract}
\medskip
If $G$ and $H$ are two cubic graphs, then an $H$-coloring of $G$ is a proper edge-coloring $f$ with edges of $H$, such that for each vertex $x$ of $G$, there is a vertex $y$ of $H$ with
$f(\partial_G(x))=\partial_H(y)$. If $G$ admits an $H$-coloring, then we will write $H\prec G$.
The Petersen coloring conjecture of Jaeger ($P_{10}$-conjecture) states that for any bridgeless cubic graph $G$, one has: $P_{10}\prec G$. The Sylvester coloring conjecture ($S_{10}$-conjecture) states that for any cubic graph $G$, $S_{10}\prec G$. In this paper, we introduce two new conjectures that are related to these conjectures. The first of them states that any cubic graph with a perfect matching admits an $S_{12}$-coloring. The second one states that any cubic graph $G$ whose edge-set can be covered with four perfect matchings, admits a $P_{12}$-coloring. We call these new conjectures $S_{12}$-conjecture and $P_{12}$-conjecture, respectively. Our first results justify the choice of graphs in $S_{12}$-conjecture and $P_{12}$-conjecture. Next, we characterize the edges of $P_{12}$ that may be fictive in a $P_{12}$-coloring of a cubic graph $G$. Finally, we relate the new conjectures to the already known conjectures by proving that $S_{12}$-conjecture implies $S_{10}$-conjecture, and $P_{12}$-conjecture and $(5,2)$-Cycle cover conjecture together imply $P_{10}$-conjecture. Our main tool for proving the latter statement is a new reformulation of $(5,2)$-Cycle cover conjecture, which states that the edge-set of any claw-free bridgeless cubic graph can be covered with four perfect matchings.
\end{abstract}


\end{frontmatter}


\section{Introduction}
\label{sec:intro}

In this paper, we consider finite, undirected graphs. They do not contain loops,
though they may contain parallel edges. We also consider pseudo-graphs, which may
contain both loops and parallel edges, and simple graphs, which contain neither loops nor parallel edges. As usual, a loop contributes to the degree of a vertex by two.

Within the frames of this paper, we assume that graphs, pseudo-graphs and simple graphs are considered up to isomorphisms. This implies that the equality $G=G'$ means that $G$ and $G'$ are isomorphic.

For a graph $G$, let $V(G)$ and $E(G)$ be the set of vertices and edges of $G$, respectively. Moreover, let $\partial_{G}(x)$ be the set of edges of $G$ that are incident to the vertex $x$ of $G$. A matching of $G$ is a set of edges of $G$ such that any two of them do not share a vertex. A matching of $G$ is perfect, if it contains $\frac{|V(G)|}{2}$ edges. A block of $G$ is a maximal $2$-connected subgraph of $G$. An end-block is a block of $G$ containing at most one vertex that is a cut-vertex of $G$. A subgraph $H$ of $G$ is even, if every vertex of $H$ has even degree in $H$. A subgraph $H$ is odd, if every vertex of $G$ has odd degree in $H$. Sometime, we will refer to odd subgraphs as joins. Observe that a perfect matching is a join of a cubic graph. A subgraph $H$ is a parity subgraph if for every vertex $v$ of $G$ $d_{G}(v)$ and $d_{H}(v)$ have the same parity. Observe that $H$ is a parity subgraph of $G$ if $G-E(H)$ is an even subgraph of $G$.



If $G$ is a cubic graph, and $K$ is a triangle in $G$, then one can obtain a cubic pseudo-graph by contracting $K$. We will denote this pseudo-graph by $G/K$. If $G/K$ is a graph, we will say that $K$ is contractible. Observe that if $K$ is not contractible, two vertices of $K$ are joined with two parallel edges, and the third vertex is incident to a bridge (see the end-blocks of the graph from Figure \ref{SylvGraph}). If $K$ is a contractible triangle, and $e$ is an edge of $K$, then let $f$ be the edge of $G$ that is incident to a vertex of $K$ and is not adjacent to $e$. $e$ and $f$ will be called opposite edges.


Let $G$ and $H$ be two cubic graphs. An $H$-coloring of $G$ is a mapping $f:E(G)\rightarrow E(H)$, such that for each vertex $x$ of $G$ there is a vertex $y$ of $H$, such that $f(\partial_{G}(x)) = \partial_{H}(y)$. If $G$ admits an $H$-coloring, then we will write $H
\prec G$. It can be easily seen that if $H\prec G$ and $K\prec H$, then $K\prec G$. In other words, $\prec$ is a transitive relation defined on the set of cubic graphs.


If $H \prec G$ and $f$ is an $H$-coloring of $G$, then for any adjacent edges $e$, $e'$
of $G$, the edges $f(e)$, $f(e')$ of $H$ are adjacent. Moreover, if the graph $H$
contains no triangle, then the converse is also true, that is, if a mapping $f : E(G) \to E(H)$
has a property that for any two adjacent edges $e$ and $e'$ of $G$, the edges $f(e)$
and $f(e')$ of $H$ are adjacent, then $f$ is an $H$-coloring of $G$.

\begin{figure}[ht]
\centering
\begin{minipage}[b]{.5\textwidth}
  \begin{center}
		\begin{tikzpicture}[style=thick]
\draw (18:2cm) -- (90:2cm) -- (162:2cm) -- (234:2cm) --
(306:2cm) -- cycle;
\draw (18:1cm) -- (162:1cm) -- (306:1cm) -- (90:1cm) --
(234:1cm) -- cycle;
\foreach \x in {18,90,162,234,306}{
\draw (\x:1cm) -- (\x:2cm);
\draw[fill=black] (\x:2cm) circle (2pt);
\draw[fill=black] (\x:1cm) circle (2pt);
}
\end{tikzpicture}
	\end{center}
	\caption{The graph $P_{10}$.}\label{PetersenGraph}
\end{minipage}%
\begin{minipage}[b]{.5\textwidth}
  
  \begin{center}
	
	\tikzstyle{every node}=[circle, draw, fill=black!50,
                        inner sep=0pt, minimum width=4pt]
	
		\begin{tikzpicture}
																							
			\node[circle,fill=black,draw] at (-5.5,-1) (n1) {};
			
			
			
			

			\node[circle,fill=black,draw] at (-6, -0.5) (n2) {};
																								
			\node[circle,fill=black,draw] at (-5,-0.5) (n3) {};
																								
			\node[circle,fill=black,draw] at (-3.5,-1) (n4) {};
																								
			\node[circle,fill=black,draw] at (-4, -0.5) (n5) {};
																								
			\node[circle,fill=black,draw] at (-3,-0.5) (n6) {};
																								
			\node[circle,fill=black,draw] at (-1.5,-1) (n7) {};
																								
			\node[circle,fill=black,draw] at (-2, -0.5) (n8) {};
																								
			\node[circle,fill=black,draw] at (-1,-0.5) (n9) {};
																								
			\node[circle,fill=black,draw] at (-3.5,-2) (n10) {};

			\path[every node]
			(n1) edge  (n2)

			edge  (n3)
			edge (n10)
																								   	
			(n2) edge (n3)
			edge [bend left] (n3)
																								       
			(n3) 
			(n4) edge (n5)
			edge (n6)
			edge (n10)
																								    
			(n5) edge (n6)
			edge [bend left] (n6)
			(n6)
																								   
			(n7) edge (n8)
			edge (n9)
			edge (n10)
																								    
			(n8) edge (n9)
			edge [bend left] (n9)
																								  
			;
		\end{tikzpicture}
																
	\end{center}
								
	\caption{The graph $S_{10}$.}
	\label{SylvGraph}
\end{minipage}
\end{figure}


Let $P_{10}$ be the well-known Petersen graph (Figure \ref{PetersenGraph}) and let $S_{10}$ be the graph from Figure \ref{SylvGraph}. $S_{10}$ is called the Sylvester graph \cite{Schrijver}. We would like to point out that usually the name “Sylvester graph” is used for a particular strongly regular
graph on $36$ vertices, and this graph should not be confused with $S_{10}$, which has $10$
vertices.

The Petersen coloring conjecture of Jaeger states:
\begin{conjecture}\label{conj:P10conjecture} (Jaeger, 1988 \cite{Jaeger}) For any bridgeless cubic
graph $G$ $P_{10} \prec G$.
\end{conjecture}
Sometimes, we will call this conjecture as $P_{10}$-conjecture. The conjecture is difficult to prove, since it can be seen that it implies the
following classical conjectures:

\begin{conjecture}
\label{conj:Berge}(Berge, unpublished) For any bridgeless cubic graph $G$ $k(G)\leq 5$.
\end{conjecture} Here $k(G)$ is the smallest number of perfect matchings that are needed to cover the edge-set of $G$

\begin{conjecture} (Berge-Fulkerson, 1972 \cite{Fulkerson,Seymour}) Any bridgeless
cubic graph $G$ contains six (not necessarily distinct) perfect matchings
$F_1, \ldots , F_6$ such that any edge of $G$ belongs to exactly two of them.
\end{conjecture} This list of six perfect matchings usually is called a Berge-Fulkerson cover of $G$.


\begin{conjecture}\label{conj:5CDC}
((5, 2)-even-subgraph-cover conjecture, \cite{Celmins1984,Preiss1981}) Any bridgeless
graph $G$ (not necessarily cubic) contains five even subgraphs such that any
edge of $G$ belongs to exactly
two of them.
\end{conjecture}

Related with the Jaeger conjecture, the following conjecture has been introduced in \cite{PetersenRemark}:
\begin{conjecture}(V. V. Mkrtchyan, 2012 \cite{PetersenRemark})
\label{conj:S10Conjecture}
For any cubic graph $G$ $S_{10} \prec G$. 
\end{conjecture}
In direct analogy with the Jaeger conjecture, we call Conjecture
\ref{conj:S10Conjecture} the Sylvester coloring conjecture or just $S_{10}$-conjecture.

A $k$-edge-coloring is an assignment of colors to edges of a graph from a set of $k$ colors such that adjacent edges receive different colors. The smallest $k$ for which a graph $G$ admits a $k$-edge-coloring is called a chromatic index of $G$ and is denoted by $\chi'(G)$. If $\alpha$ is a $k$-edge-coloring of a cubic graph $G$, then an edge $e=uv$ is called poor (rich) in $\alpha$, if the five edges of $G$ incident to $u$ or $v$ are colored with three (five) colors. $\alpha$ is called a normal $k$-edge-coloring of $G$ if any edge of $G$ is either poor or rich in $\alpha$. Observe that not all cubic graphs admit a normal $k$-edge-coloring for some $k$. An example of such a graph is the Sylvester graph (Figure \ref{SylvGraph}). On the positive side, all simple cubic graphs admit a normal $7$-edge-coloring \cite{NormalCubics7}. The smallest $k$ (if it exists) for which a cubic graph $G$ admits a normal $k$-edge-coloring is called a normal chromatic index of $G$ and is denoted by $\chi'_N(G)$.

Normal colorings were introduced by Jaeger in \cite{Jaeger1985}, where among other results, he showed that for a cubic graph $G$ $\chi'_N(G)\leq 5$ if and only if $G$ admits a $P_{10}$-coloring. This allowed him to obtain a reformulation of Conjecture \ref{conj:P10conjecture}, which states that for any bridgeless cubic graph $G$ $\chi'_N(G)\leq 5$.



In this paper, we introduce two new conjectures that are related to Conjectures \ref{conj:P10conjecture} and \ref{conj:S10Conjecture}. In section \ref{sec:aux}, we discuss some auxiliary results that will be useful later in the paper. In section \ref{sec:hereditary}, we briefly discuss so-called hereditary classes of cubic graphs that allowed us to come up with these two new conjectures. Then in section \ref{sec:main}, we present our main results. Finally, in section \ref{sec:future}, we discuss some open problems. Terms and concepts that we do not define in the paper can be found in \cite{West,Zhang1997}.

\section{Auxiliary results}
\label{sec:aux}

In this section, we present some auxiliary results that will be useful later. Our first result is a lemma about some properties of $H$-colorings of cubic graphs. Though all these properties are known before, for the sake of completeness we give complete proofs.

\begin{lemma}\label{lem:propHcolorings}
Suppose that $G$ and $H$ are cubic graphs with $H \prec G$, and let $f$ be an $H$-coloring of
$G$. Then:
\begin{enumerate}[(a)]
\item If $M$ is any matching of $H$, then $f^{-1}(M)$ is a matching of $G$;

\item $\chi'(G) \leq \chi'(H)$;

\item If $M$ is a perfect matching of $H$, then $f^{-1}(M)$ is a perfect matching of $G$;

\item $k(G)\leq k(H)$;

\item If $H$ admits a Berge-Fulkerson cover, then $G$ also admits a Berge-Fulkerson cover;

\item For every even subgraph $C$ of $H$, $f^{-1}(C)$ is an even subgraph of $G$;

\item For every bridge $e$ of $G$, the edge $f(e)$ is a bridge of $H$.

\item If $H$ is bridgeless, then $G$ is bridgeless as well;

\item $\chi'_N(G)\leq \chi'_N(H)$.

\end{enumerate}
\end{lemma}

\begin{proof}
(a) The proof of this statement follows from the definition of $H$-coloring: as adjacent edges of $G$ must be colored with adjacent edges of $H$, then clearly the pre-image of a matching in $H$ must be a matching in $G$.

(b) Assume that $\chi'(H)=s$ and let $M_1,..., M_s$ be the color classes of $H$ in an $s$-edge-coloring. Consider $f^{-1}(M_1)$,...,$f^{-1}(M_s)$. Observe that due to (a), they are forming $s$ matchings covering the edge-set of $G$. Thus, $\chi'(G) \leq s= \chi'(H)$.

(c) Let $M$ be a perfect matching of $H$. Then by (a), $f^{-1}(M)$ is a matching of $G$. Let us show that it covers all vertices of $G$. Let $v$ be a vertex of $G$. Then the three edges incident to $v$ are colored by a similar three edges of $H$. Since $M$ is a perfect matching of $H$, one of these three edges must belong to $M$, hence $f^{-1}(M)\cap \partial_G(v)\neq \emptyset$. Thus, $f^{-1}(M)$ is a perfect matching of $G$.

(d) The proof of this is similar to that of (b): let $k(H)=s$ and let $M_1,..., M_s$ be the $s$ perfect matchings of $H$ covering $E(H)$. Consider $f^{-1}(M_1)$,...,$f^{-1}(M_s)$. Observe that due to (c), they are forming $s$ perfect matchings covering the edge-set of $G$. Thus, $k(G) \leq s= k(H)$.

(e) Let $C=(F_1,...,F_6)$ be a Berge-Fulkerson cover of $H$. Consider the list $f^{-1}(C)=(f^{-1}(F_1),...,f^{-1}(F_6))$. Observe that due to (c) they are forming a list of six perfect matchings of $G$. Moreover, every edge of $G$ belongs to at least two of these perfect matchings. Hence $f^{-1}(C)$ is a Berge-Fulkerson cover of $G$.

(f) Let $C$ be an even subgraph of $H$. Let us show that any vertex $v$ of $G$ has even degree in $f^{-1}(C)$. Since $H$ is cubic, $C$ is a disjoint union of cycles. Assume that in $f$ the three edges incident to $v$ are colored with three edges incident to a vertex $w$ of $H$. Then if $w$ is isolated in $C$, then clearly $v$ is isolated in $f^{-1}(C)$. On the other hand, if $w$ has degree two in $C$, then $v$ is of degree two in $f^{-1}(C)$. Thus, $v$ always has even degree in $f^{-1}(C)$, or $f^{-1}(C)$ is an even subgraph of $G$.

(g) Let $e$ be a bridge of $G$ and let $(X, V(G)\backslash X)$ be a partition of $V(G)$, such that $\partial_G(X)=\{e\}$. Assume that the edge $f(e)$ is not a bridge in $H$. Then there is a cycle $C$ in $H$ that contains the edge $f(e)$. By (f) $f^{-1}(C)$ is an even subgraph of $G$ that has non-empty intersection with $\partial_G(X)$. Since the intersection of an even subgraph with $\partial_G(X)$ always contains an even number of edges, we have that $\partial_G(X)$ contains at least two edges which contradicts our assumption.

(h) This follows from (g): if $H$ has no a bridge, then any edge of $G$ cannot be a bridge, as otherwise its color in $f$ will be a bridge in $H$.

(i) Assume that $\chi'_N(H)=s$, and let $g$ be a normal $s$-edge-coloring of $H$. Consider a mapping $h$ defined on the edge-set of $G$ as follows: for any edge $e$ of $G$, let $h(e)=g(f(e))$. Let us show that $h$ is a normal $s$-edge-coloring of $G$. Let $e=vw$ be any edge of $G$. Assume that in $f$ the three edges incident to $v$ are colored by the three edges incident to a vertex $u_1$ of $H$, the three edges incident to $w$ are colored by the three edges incident to a vertex $u_2$ of $H$. 

If $u_1=u_2$, then the edge $e$ is poor in $h$. Thus, we can assume that $u_1\neq u_2$. Since $e\in \partial_G(v)\cap \partial_G(w)$, we have that $u_1u_2\in E(H)$ and $f(e)=u_1u_2$. Now, observe that since $g$ is a normal edge-coloring, we have that if $f(e)$ is a poor edge in $g$, then $e$ is a poor edge in $h$, and if $f(e)$ is a rich edge in $g$, then $e$ is a rich edge in $h$. Thus, $h$ is a normal $s$-edge-coloring of $G$. Hence $\chi'_N(G)\leq s=\chi'_N(H)$. The proof of the lemma is complete.
\end{proof}

\bigskip

We will need some results on claw-free bridgeless cubic graphs. Recall that a graph $G$ is claw-free, if it does not contain $4$ vertices, such that the subgraph of $G$ induced on these vertices is isomorphic to $K_{1,3}$. In \cite{ChudSeyClawFreeChar}, arbitrary claw-free graphs are characterized. In \cite{sang-il_oum:2011}, Oum has characterized simple, claw-free bridgeless cubic graphs. In order to formulate Oum's result, we need some definitions. In a claw-free simple cubic graph $G$ any vertex belongs to $1$, $2$, or $3$ triangles. If a vertex $v$ belongs to $3$ triangles of $G$, then the component of $G$ containing $v$ is isomorphic to $K_4$ (Figure \ref{fig:K4}). An induced subgraph of $G$ that is isomorphic to $K_4-e$ is called a diamond \cite{sang-il_oum:2011}. It can be easily checked that in a claw-free cubic graph no $2$ diamonds intersect.

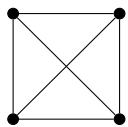
\begin{figure}[!htbp]
\begin{center}
\begin{tikzpicture}[scale=0.35]

 \tikzstyle{every node}=[circle, draw, fill=black!,
                        inner sep=0pt, minimum width=4pt]
  
  \node[circle,fill=black,draw] at (-2,2) (n2) {};

   \node[circle,fill=black,draw] at (-2,-2) (n3) {};

    \node[circle,fill=black,draw] at (-6,2) (n4) {};
    
    \node[circle,fill=black,draw] at (-6,-2) (n5) {};

 \draw (n2)--(n3);
 
 \draw (n2)--(n4);
 \draw (n3)--(n4);
 \draw (n3)--(n5);
 \draw (n2)--(n5);
 
 \draw (n4)--(n5);

\end{tikzpicture}
\end{center}
\caption{The graph $K_4$.} \label{fig:K4}
\end{figure}

A string of diamonds of $G$ is a maximal sequence $F_{1},...,F_{k}$ of diamonds, in which $F_{i}$ has a vertex adjacent to a vertex of $F_{i+1}$, $1\leq i \leq k-1$.  A string of diamonds has exactly $2$ vertices of degree $2$, which are called the head and the tail of the string. Replacing an edge $e = uv$ with a string of diamonds with the head $x$ and the tail $y$ is to remove $e$ and add edges $(u,x)$ and $(v,y)$.

If $G$ is a connected claw-free simple cubic graph such that each vertex lies in a diamond, then $G$ is called a ring of diamonds. It can be easily checked that each vertex of a ring of diamonds lies in exactly $2$ diamonds. As in \cite{sang-il_oum:2011}, we require that a ring of diamonds contains at least $2$ diamonds.

\begin{proposition}\label{prop:OumClawfreebridgelessCharac} (Oum, \cite{sang-il_oum:2011}) $G$ is a connected claw-free simple bridgeless cubic graph, if and only if
\begin{enumerate}
    \item [(1)] $G$ is isomorphic to $K_4$, or
    
    \item [(2)] $G$ is a ring of diamonds, or
    
    \item [(3)] there is a connected bridgeless cubic graph $H$, such that $G$ can be obtained from $H$ by replacing some edges of $H$ with strings of diamonds, and by replacing any vertex of $H$ with a triangle.
\end{enumerate}
\end{proposition}

The next auxiliary result allows us to relate coverings with even subgraphs to coverings with specific parity subgraphs.

\begin{theorem}
\label{thm:CQthm} (Theorem 3.3 in \cite{CQJoins}) For a graph $G$, the following two conditions are equivalent:
\begin{enumerate}[(1)]
    \item $G$ contains five even subgraphs such that any edge of $G$ belongs to exactly two of them;
    
    \item $G$ contains four parity subgraphs such that each edge belongs to  either one or two of the parity subgraphs.
\end{enumerate}
\end{theorem}

Conjecture \ref{conj:Berge} implies that for any claw-free bridgeless cubic graph $G$ $k(G)\leq 5$. We are unable to find an example attaining this bound. We suspect that:

\begin{conjecture}
\label{conj:ClawFreeBerge} For any claw-free bridgeless cubic graph $G$ $k(G)\leq 4$.
\end{conjecture}

Our final auxiliary result is a theorem proved by Giuseppe Mazzuoccolo which offers a new reformulation of Conjecture \ref{conj:5CDC}.

\begin{theorem}
\label{thm:Giuseppe5CDCClawFreeEq} Conjectures \ref{conj:5CDC} and \ref{conj:ClawFreeBerge} are equivalent.
\end{theorem}

\begin{proof} Assume that Conjecture \ref{conj:ClawFreeBerge} is true. Let us show that Conjecture \ref{conj:5CDC} is also true. It is known that it suffices to prove Conjecture \ref{conj:5CDC} for cubic graphs \cite{Zhang1997}. Let $G$ be an arbitrary bridgeless cubic graph. Consider the cubic graph $H$ obtained from $G$ by replacing every vertex of $G$ with a triangle. Observe that $H$ is a claw-free bridgeless cubic graph. By Conjecture \ref{conj:ClawFreeBerge}, the edge-set of $H$ can be covered with four perfect matchings. Observe that perfect matchings are parity subgraphs in cubic graphs, hence by Theorem \ref{thm:CQthm}, $H$ admits a list of $5$ even subgraphs covering each edge exactly twice.

In order to complete the proof, let us observe that if a cubic graph $K$ admits a list of $5$ even subgraphs covering each edge exactly twice and it contains a triangle $T$, then the graph $K/T$ also admits a list of $5$ even subgraphs covering each edge exactly twice. In order to see this, let $\mathcal{C}=(Ev_1,...,Ev_5)$ be the list of $5$ even subgraphs covering the edges of $K$ twice. Then it is easy to see that the edges of the $3$-cut $\partial_K(T)$ are covered as follows: first edge belongs to $Ev_1$ and $Ev_2$, the second edge belongs to $Ev_1$ and $Ev_3$, and finally the third edge belongs to $Ev_2$ and $Ev_3$. Moreover, $Ev_4$ and $Ev_5$ do not intersect the $3$-cut. One can always achieve this by renaming the even subgraphs. Now, if we consider the restrictions of $(Ev_1,...,Ev_5)$ to $K/T$, we will have that they are forming a list of $5$ even subgraphs covering each edge of $K/T$ exactly twice. 

Applying this observation $|V(G)|$ times to $H$, we will get the statement for the original cubic graph $G$.

For the proof of the converse statement, let us assume that Conjecture \ref{conj:5CDC} is true, and show that any claw-free bridgeless cubic graph $G$ can be covered with four perfect matchings. We prove the latter statement by induction on $|V(G)|$. If $|V(G)|=2$, the the statement is trivially true. Assume that it is true for all claw-free bridgeless cubic graphs with less $n$ vertices and let us consider a claw-free bridgeless cubic graph $G$ with $n\geq 4$ vertices.

Clearly, we can assume that $G$ is connected. Let us show that we can assume that $G$ is simple. On the opposite assumption, consider the vertices $u$ and $v$ that are joined with two edges. Let $u'$ and $v'$ be the the other neighbors of $u$ and $v$, respectively. Consider a cubic graph $G'$ obtained from $G-\{u,v\}$ by adding a possibly parallel edge $u'v'$. Observe that $G'$ is a bridgeless cubic graph with $|V(G')|<n$. Moreover, it is claw-free. Thus, by induction hypothesis, $G'$ can be covered with four perfect matchings. Now, it is easy to see that using these list of four perfect matchings of $G'$ we can construct a list of four perfect matchings of $G$ covering $E(G)$. 

Thus, without loss of generality, we can assume that $G$ is simple. Hence, we can apply Proposition \ref{prop:OumClawfreebridgelessCharac}. If $G$ meets the first or the second condition of the proposition, then it is easy to see that $G$ is $3$-edge-colorable, hence it can be covered with three perfect matchings. Thus, we can assume that there is a connected bridgeless cubic graph $H$ such that $G$ can be obtained from $H$ by replacing some edges of $H$ with strings of diamonds and every vertex of $H$ with a triangle.  

Let us show that we can also assume that $G$ has no string of diamonds. Assume it has one. Let it be $S$ whose head and tails are $u$ and $v$, respectively. Let $u'$ and $v'$ be the neighbors of $u$ and $v$, respectively, that lie outside $S$. Consider a graph $G'$ obtained from $G-V(S)$ by adding a possibly parallel edge $u'v'$. Observe that $G'$ is a bridgeless cubic graph with $|V(G')|<n$. Moreover, it is claw-free. Thus, by induction hypothesis, $G'$ can be covered with four perfect matchings. Now, it is easy to see that using these list of four perfect matchings of $G'$ we can construct a list of four perfect matchings of $G$ covering $E(G)$.

Thus, without loss of generality, we can assume that $G$ can be obtained from the connected bridgeless cubic graph $H$ by replacing its every vertex with a triangle. By Conjecture \ref{conj:5CDC}, $H$ has a list of five even subgraphs covering its edges exactly twice. By Theorem \ref{thm:CQthm}, we have that $H$ admits a cover with four joins such that each edge of $H$ is covered once or twice. Let $v$ any vertex of $H$ and let $C=(T_1, T_2, T_3, T_4)$ be the cover of $H$ with four joins. Since each edge of $H$ is covered once or twice in $C$, we have that there is at most one join in $C$ that contains all three edges incident to $v$. Thus, for any vertex $v$ we have that either one of joins contains all three edges incident to $v$ and the other three joins contain exactly one edge incident to $v$, or all joins contain exactly one edge incident to $v$. Now, it is not hard to see that these four joins covering $H$ can be extended to four perfect matchings of $G$ so that they cover $G$. The proof of the theorem is complete.
\end{proof}

\section{Hereditary classes of cubic graphs}
\label{sec:hereditary}

In this section, we briefly discuss hereditary classes of cubic graphs. It is these classes that allowed us to come up with more conjectures related to Conjectures \ref{conj:P10conjecture} and \ref{conj:S10Conjecture}. 

If $G$ and $H$ are two cubic graphs with $H\prec G$ or $G\prec H$, then we will say that $G$ and $H$ are comparable. A (not necessarily finite) set of cubic graphs is said to be an anti-chain, if any two cubic graphs from the set are not comparable. Let $\mathcal{C}$ be the class of all connected cubic graphs. If $\mathcal{M}\subseteq \mathcal{C}$ is a class of connected cubic graphs, then we will say that $\mathcal{M}$ is hereditary, if for any cubic graphs $G$ and $H$, if $H\in \mathcal{M}$ and $H\prec G$, then $G\in \mathcal{M}$. Assume that $\mathcal{B}\subseteq \mathcal{M}$ is a subset of some hereditary class $\mathcal{M}$ of cubic graphs. We will say that $\mathcal{B}$ is a basis for $\mathcal{M}$, if $\mathcal{B}$ is an anti-chain and for any connected cubic graph $G$ we have that $G\in \mathcal{M}$ if and only if there is a cubic graph $H\in \mathcal{B}$, such that $H\prec G$.

Our starting question is the following: does every hereditary class of cubic graphs admit a finite basis, that is, a basis with finitely many elements? It turns out that the answer to this question is negative. Let $\mathcal{I}$ be the infinite anti-chain of cubic graphs constructed in \cite{Samal2017}. Consider the class $\mathcal{M}$ of connected cubic graphs $G$, such that for any $G$ we have: $G\in \mathcal{M}$, if and only if there is a cubic graph $H\in \mathcal{I}$, such that $H\prec G$. It is easy to see that $\mathcal{M}$ is a hereditary class of cubic graphs without a finite basis.

Despite this, one may still look for interesting hereditary classes arising in Graph theory, that admit a finite basis. Below, we discuss some such classes. The first class is $\mathcal{C}$-the class of all connected cubic graphs. Clearly, it is hereditary. Observe that any connected cubic graph admitting an $S_{10}$-coloring belongs to $\mathcal{C}$. On the other hand, Conjecture \ref{conj:S10Conjecture} predicts that any cubic graph from $\mathcal{C}$ admits an $S_{10}$-coloring. Thus, we can view Conjecture \ref{conj:S10Conjecture} as a statement that $S_{10}$ forms a basis for $\mathcal{C}$.

Let $\mathcal{C}_b$ be the class of all connected bridgeless cubic graphs. (h) of Lemma \ref{lem:propHcolorings} implies that $\mathcal{C}_b$ is a hereditary class of cubic graphs. Observe that any connected cubic graph admitting a $P_{10}$-coloring belongs to $\mathcal{C}_b$. On the other hand, Conjecture \ref{conj:P10conjecture} predicts that any bridgeless cubic graph from $\mathcal{C}_b$ admits a $P_{10}$-coloring. Thus, we can view Conjecture \ref{conj:P10conjecture} as a statement that $P_{10}$ forms a basis for $\mathcal{C}_b$.

Let $\mathcal{C}_3$ be the class of all connected 3-edge-colorable cubic graphs. (b) of Lemma \ref{lem:propHcolorings} implies that $\mathcal{C}_3$ is a hereditary class of cubic graphs. Let $H$ be any connected 3-edge-colorable cubic graph. Observe that any cubic graph $G$ is 3-edge-colorable if and only if $H\prec G$. Thus, $H$ forms a basis for $\mathcal{C}_3$.

\begin{figure}[ht]
\centering
\begin{minipage}[b]{.5\textwidth}
  \begin{center}
	\tikzstyle{every node}=[circle, draw, fill=black!50,
                        inner sep=0pt, minimum width=4pt]
                        
		\begin{tikzpicture}
																								
			\node[circle,fill=black,draw] at (-5.5,-1) (n1) {};
																								
			\node[circle,fill=black,draw] at (-6, -0.5) (n2) {};
																								
			\node[circle,fill=black,draw] at (-5,-0.5) (n3) {};
																								
			\node[circle,fill=black,draw] at (-3.5,-1) (n4) {};
																								
			\node[circle,fill=black,draw] at (-4, -0.5) (n5) {};
																								
			\node[circle,fill=black,draw] at (-3,-0.5) (n6) {};
																								
			\node[circle,fill=black,draw] at (-1.5,-1) (n7) {};
																								
			\node[circle,fill=black,draw] at (-2, -0.5) (n8) {};
																								
			\node[circle,fill=black,draw] at (-1,-0.5) (n9) {};
																								
			\node[circle,fill=black,draw] at (-5.5,-2) (n10) {};
																								
			\node[circle,fill=black,draw] at (-3.5,-2) (n11) {};
																								 
			\node[circle,fill=black,draw] at (-1.5,-2) (n12) {};
																								
			\path[every node]
			(n1) edge  (n2)
																								    
			edge  (n3)
			edge (n10)
																								   	
			(n2) edge (n3)
			edge [bend left] (n3)
																								       
			(n3) 
			(n4) edge (n5)
			edge (n6)
			edge (n11)
																								    
			(n5) edge (n6)
			edge [bend left] (n6)
			(n6)
																								   
			(n7) edge (n8)
			edge (n9)
			edge (n12)
																								    
			(n8) edge (n9)
			edge [bend left] (n9)
																								   
			(n10) edge (n11)
			edge (n12)
																								   
			(n10) edge [bend right] (n12)
																								  
			;
		\end{tikzpicture}
																
	\end{center}
	
	\caption{The graph $S_{12}$.}\label{fig:Sylvester12}
\end{minipage}%
\begin{minipage}[b]{.5\textwidth}
  	\begin{center}
	
	\tikzstyle{every node}=[circle, draw, fill=black!50,
                        inner sep=0pt, minimum width=4pt]
                        
		\begin{tikzpicture}[scale=0.75]
																								
			\node[circle,fill=black,draw] at (0,0) (n1) {};
																								
			\node[circle,fill=black,draw] at (-2.5, 0) (n2) {};
																								
			\node[circle,fill=black,draw] at (2.5,0) (n3) {};
																								
			\node[circle,fill=black,draw] at (-1.5,-0.75) (n4) {};
																								
			\node[circle,fill=black,draw] at (1.5, -0.75) (n5) {};
																								
			\node[circle,fill=black,draw] at (-1,-2) (n6) {};
																								
			\node[circle,fill=black,draw] at (1,-2) (n7) {};
																								
			\node[circle,fill=black,draw] at (-2, -3) (n8) {};
																								
			\node[circle,fill=black,draw] at (2,-3) (n9) {};
			
			
			\node[circle,fill=black,draw] at (-1,2) (n01) {};
			
			\node[circle,fill=black,draw] at (0,1) (n02) {};
			
			\node[circle,fill=black,draw] at (1,2) (n03) {};

			\path[every node]
			(n1) edge  (n6)
																								    
			edge  (n7)
																								   	
			(n2) edge (n4)
			edge (n8)
																								       
			(n3) edge (n5)
			edge (n9)
			(n4) edge (n5)
			edge (n7)

			(n5) edge (n6)
																								   
			(n6)edge (n8)
																								   
			(n7) edge (n9)

			(n8) edge (n9)
			
			(n01) edge (n02)
			      edge (n03)
			      edge (n2)
			 (n02) edge (n03)
			        edge (n1)
			 (n03) edge (n3)
			;

		\end{tikzpicture}
																
	\end{center}
	\caption{The graph $P_{12}$.}\label{fig:Petersen12}
\end{minipage}
\end{figure}

Let $\mathcal{C}_p$ be the class of all connected cubic graphs containing a perfect matching. (c) of Lemma \ref{lem:propHcolorings} implies that $\mathcal{C}_p$ is a hereditary class of cubic graphs. Observe that any connected cubic graph admitting an $S_{12}$-coloring (the graph from Figure \ref{fig:Sylvester12}) belongs to $\mathcal{C}_p$. On the other hand, we suspect that

\begin{conjecture}
\label{conj:S12conjecture} Any cubic graph with a perfect matching admits an $S_{12}$-coloring.
\end{conjecture}

Conjecture \ref{conj:S12conjecture} predicts that all cubic graphs from $\mathcal{C}_p$ admit an $S_{12}$-coloring. Thus, we can view Conjecture \ref{conj:S12conjecture} as a statement that $S_{12}$ forms a basis for $\mathcal{C}_p$. Let us note that Conjecture \ref{conj:S12conjecture} has been verified for claw-free cubic graphs in \cite{AnushSylvester}.

Let $\mathcal{C}^{(4)}$ be the class of all connected cubic graphs $G$ with $k(G)\leq 4$. (d) of Lemma \ref{lem:propHcolorings} implies that $\mathcal{C}^{(4)}$ is a hereditary class of cubic graphs. Observe that any connected cubic graph admitting a $P_{12}$-coloring (the graph from Figure \ref{fig:Sylvester12}) belongs to $\mathcal{C}^{(4)}$. On the other hand, we suspect that

\begin{conjecture}
\label{conj:P12conjecture} Any cubic graph $G$ with $k(G)\leq 4$ admits a $P_{12}$-coloring.
\end{conjecture}

Conjecture \ref{conj:P12conjecture} predicts that all cubic graphs from $\mathcal{C}^{(4)}$ admit a $P_{12}$-coloring. Thus, we can view Conjecture \ref{conj:P12conjecture} as a statement that $P_{12}$ forms a basis for $\mathcal{C}^{(4)}$. Also, note that (e) of Lemma \ref{lem:propHcolorings} implies that Conjecture 4.9 from \cite{CQJoins} is a consequence of Conjecture \ref{conj:P12conjecture}.

\section{The main results}
\label{sec:main}

In this section, we obtain our main results. First, we discuss the choice of graphs $P_{12}$ and $S_{12}$ in Conjectures \ref{conj:P12conjecture} and \ref{conj:S12conjecture}, respectively. For this purpose, we recall the following two theorems that are proved in \cite{PetersenRemark}.

\begin{theorem}\label{thm:Petersen}
If $G$ is a connected bridgeless cubic graph with $G \prec P_{10}$, then $G = P_{10}$.
\end{theorem}

\begin{theorem}\label{thm:Sylvester}
If $G$ is a connected cubic graph with $G \prec S_{10}$, then $G = S_{10}$.
\end{theorem}

The first theorem suggests that in Conjecture \ref{conj:P10conjecture} the graph $P_{10}$ cannot be replaced with any other connected bridgeless cubic graph. Similarly, the second theorem suggests that in Conjecture \ref{conj:S10Conjecture}  the graph $S_{10}$ cannot be replaced with other connected cubic graph. Now, we are going to obtain similar results for Conjectures \ref{conj:P12conjecture} and \ref{conj:S12conjecture}.

\begin{theorem}\label{thm:P10P12thm}
Let $G$ be a connected bridgeless cubic graph with $G \prec P_{12}$. Then either $G=P_{10}$ or $G= P_{12}$. 
\end{theorem}

\begin{proof}
Assume that $f$ is a $G$-coloring of $P_{12}$. Consider the triangle $T$ in $P_{12}$. Assume that the edges of $T$ are $e_1, e_2, e_3$. Since these three edges are pairwise adjacent in $P_{12}$, we have that the edges $f(e_1), f(e_2), f(e_3)$ are pairwise adjacent in $G$. We have two cases to consider:

\medskip

Case 1: There is a vertex $v$ of $G$, such that $\partial_G(v)=\{f(e_1), f(e_2), f(e_3)\}$. Observe that in this case the edges of the 3-edge-cut $\partial_{P_{12}}(V(T))$ are colored by $f(e_1), f(e_2), f(e_3)$, respectively. Thus, if we contract $T$ in $P_{12}$, we will get a $G$-coloring of $P_{10}$. Hence, by Theorem \ref{thm:Petersen}, $G=P_{10}$.

\medskip

Case 2: The edges $f(e_1), f(e_2), f(e_3)$ form a triangle $T_0$ in $G$. Observe that in this case the edges of the 3-edge-cut $\partial_{P_{12}}(V(T))$ are colored by the edges of the 3-edge-cut $\partial_{G}(V(T_0))$. Thus, $f$ induces a $G/T_{0}$-coloring of $P_{12}/T=P_{10}$. Hence, by Theorem \ref{thm:Petersen}, $G/T_0=P_{10}$, which implies that $G= P_{12}$. The proof of the theorem is complete.
\end{proof}

\begin{corollary}
\label{cor:P12corollary} Let $G$ be a connected bridgeless cubic graph with $k(G)\leq 4$ and $G\prec P_{12}$. Then $G= P_{12}$.
\end{corollary}

\begin{theorem}\label{thm:S10S12thm}
Let $G$ be a connected cubic graph with $G \prec S_{12}$. Then either $G=S_{10}$ or $G= S_{12}$. 
\end{theorem}

\begin{proof} Assume that $f$ is a $G$-coloring of $S_{12}$. Consider the central triangle $T$ in $S_{12}$, that is, the unique triangle $T$ such that all edges of $\partial_{S_{12}}(V(T))$ are bridges. Assume that the edges of $T$ are $e_1, e_2, e_3$. Since these three edges are pairwise adjacent in $S_{12}$, we have that the edges $f(e_1), f(e_2), f(e_3)$ are pairwise adjacent in $G$. We have two cases to consider:

\medskip

Case 1: There is a vertex $v$ of $G$, such that $\partial_G(v)=\{f(e_1), f(e_2), f(e_3)\}$. Observe that in this case the edges of the 3-edge-cut $\partial_{S_{12}}(V(T))$ are colored by $f(e_1), f(e_2), f(e_3)$, respectively. Thus, if we contract $T$ in $S_{12}$, we will get a $G$-coloring of $S_{10}$. Hence, by Theorem \ref{thm:Sylvester}, $G=S_{10}$.

\medskip

Case 2: The edges $f(e_1), f(e_2), f(e_3)$ form a triangle $T_0$ in $G$. Observe that in this case the edges of the 3-edge-cut $\partial_{S_{12}}(V(T))$ are colored by the edges of the 3-edge-cut $\partial_{G}(V(T_0))$. Moreover, since all edges of $\partial_{S_{12}}(V(T))$ are bridges, by (g) of Lemma \ref{lem:propHcolorings}, we have that the three edges of $\partial_{G}(V(T_0))$ are bridges. This, in particular, means that $T_0$ is a contractible triangle in $G$. Observe that $f$ induces a $G/T_{0}$-coloring of $S_{12}/T=S_{10}$. Hence, by Theorem \ref{thm:Sylvester}, $G/T_0=S_{10}$. Moreover, the new vertex $v_{T_0}$ of $G/T_0$ corresponding to $T_0$, is incident to three bridges. Hence $v_{T_0}$ is the unique cut-vertex of $G/T_0=S_{10}$ that is incident to three bridges. This means that $G=S_{12}$. The proof of the theorem is complete.
\end{proof}

\begin{corollary}
\label{cor:S12corollary} Let $G$ be a connected cubic graph with a perfect matching such that $G\prec S_{12}$. Then $G= S_{12}$.
\end{corollary}

In the next statement, we discuss the following problem: assume that a bridgeless cubic $G$ graph admits a $P_{10}$-coloring such that one of the edges of $P_{10}$ is not used. What can we say about $G$? We discuss the related problem for Conjecture \ref{conj:P12conjecture} afterwards. Let us note that the following statement is proved by Eckhard Steffen.

\begin{proposition}
\label{prop:P10prop} Let $G$ be a bridgeless cubic that admits a $P_{10}$-coloring $f$, such that for an edge $e$ of $P_{10}$, we have: $f^{-1}(e)=\emptyset$. Then $\chi'(G)=3$.
\end{proposition}

\begin{proof} (\cite{SteffenP}) Assume that the edge $e$ of $P_{10}$ is not used in a $P_{10}$-coloring $f$ of $G$. We have that there exist two perfect matchings $M_1$ and $M_2$ of $P_{10}$, such that $M_1\cap M_2=\{e\}$. By (c) of Lemma \ref{lem:propHcolorings}, we have that $f^{-1}(M_1)$ and $f^{-1}(M_2)$ are perfect matchings in $G$. Since the edge $e$ is not used in $f$, we have that the perfect matchings are edge-disjoint in $G$. Thus $\chi'(G)=3$. The proof is complete.
\end{proof}

Next, we characterize the edges of $P_{12}$, which can be fictive in a $P_{12}$-coloring of a graph with $k(G)\leq 4$.

\begin{proposition} Let $G$ be a bridgeless cubic graph and let $T$ be the unique triangle of $P_{12}$.
\begin{enumerate}[(a)]
    \item If $G$ admits a $P_{12}$-coloring $f$, such that for an edge $e\notin T$ of $P_{12}$, we have that $f^{-1}(e)=\emptyset$, then $\chi'(G)=3$.
    
    \item There exist infinitely many bridgeless cubic graphs $G$ with $k(G)=4$, such that $G$ admits a $P_{12}$-coloring $f$, such that for any edge $e\in T$, we have: $f^{-1}(e)=\emptyset$.
\end{enumerate}
\label{prop:P12prop}
\end{proposition}

\begin{proof} (a) We follow the approach of the proof of Proposition \ref{prop:P10prop}, that is, we find two perfect matchings of $P_{12}$ whose intersection is $e$. Assume that $e\notin T$. 

If $e\notin \partial_{P_{12}}(V(T))$, then we have that there exist two perfect matchings $M_1$ and $M_2$ of $P_{10}$, such that $M_1\cap M_2=\{e\}$. Now, these two perfect matchings can be uniquely extended to perfect matchings $N_1$ and $N_2$ of $P_{12}$. Observe that $N_1\cap N_2=\{e\}$.

On the other hand, if $e\in \partial_{P_{12}}(V(T))$, then one can find a perfect matching $N_1$ intersecting $\partial_{P_{12}}(V(T))$ in a single edge and a perfect matching $N_2$ intersecting $\partial_{P_{12}}(V(T))$ in three edges, such that $N_1\cap N_2=\{e\}$.

(b) Start with arbitrary 3-edge-colorable cubic graph $H$ and consider the cubic graph $G$ obtained from $H$ by replacing every vertex of $G$ with $P_{10}-v$. Since $P_{10}-v$ is not 3-edge-colorable, we have that $G$ is not $3$-edge-colorable, hence $k(G)\geq 4$. Let us show that we have equality here. Consider the three edges incident to $v$, and let it be our colors in a 3-edge-coloring of $H$. Now, color the remaining copies of $P_{10}-v$ in $G$ by edges of $P_{10}-v$, so that each edge is colored with its copy. As a result, we get a $P_{12}$-coloring of $G$. Thus, by (d) of Lemma \ref{lem:propHcolorings}, we have $k(G)\leq 4$. Hence $k(G)=4$. Moreover, in the $P_{12}$-coloring of $G$ the edges of $T$ are not used. The proof is complete.
\end{proof}

In the final part of the paper we establish some connections among the discussed conjectures.

\begin{theorem} Conjecture \ref{conj:S12conjecture} implies Conjecture \ref{conj:S10Conjecture}.
\end{theorem}

\begin{proof} Assume that Conjecture \ref{conj:S12conjecture} is true. We claim that any cubic graph $G$ admits an $S_{10}$-coloring. In this proof, we will assume the following notation for the edges of $S_{12}$: the three bridges of $S_{12}$ are denoted by $a, b, c$, the edges of the unique contractible triangle of $S_{12}$ are denoted by $x,y,z$, such that $x$ and $a$, $y$ and $b$, $z$ and $c$ are opposite edges. Finally, the edges of the end-block containing a vertex incident to $a$ have the following labels: the edges incident to $a$ are $a_1$ and $a_2$, and the parallel edges are $a_3$ and $a_4$. Similarly, we label other edges by $b_1, b_2, b_3, b_4$ and $c_1, c_2, c_3, c_4$.

Let $G$ be a cubic graph. If $G$ contains a perfect matching, then it has an $S_{12}$-coloring. Since $S_{12}$ has an $S_{10}$-coloring, we have the statement in this case. Thus, without loss of generality, we can assume that $G$ does not contain a perfect matching.

Consider the graph $G_{\Delta}$ obtained from $G$ by replacing all vertices of $G$ with triangles. Observe that $G_{\Delta}$ contains a perfect matching. An example of such a matching would be $E(G)$.

Thus, there exists a smallest subset $U\subseteq V(G)$, such that if we replace all vertices of $U$ with triangles, we will get a cubic graph $H$ containing a perfect matching. 

By Conjecture \ref{conj:S12conjecture}, $H$ admits an $S_{12}$-coloring $f$. Now, we claim that all triangles of $H$ corresponding to vertices of $U$ are colored with triangles of $S_{12}$.

Assume the opposite, that is, there is a triangle $T$ corresponding to a vertex of $H$, such that $f(E(T))=\partial_{S_{12}}(v)$ for some vertex $v$ of $S_{12}$. Consider the graph $H'$ obtained from $H$ by contracting $T$. Observe that the resulting graph $H'$ still has an $S_{12}$-coloring, hence by (c) of Lemma \ref{lem:propHcolorings} it contains a perfect matching. However, this violates the definition of the set $U$, since we found a smaller subset of vertices, whose replacement with triangles was leading to a cubic graph containing a perfect matching.

Now, all triangles of $H$ corresponding to vertices of $U$ are colored with triangles of $S_{12}$. Let us show that all these triangles corresponding to vertices of $U$ are colored with the central triangle of $S_{12}$, that is the only contractible triangle of $S_{12}$.

On the opposite assumption, assume that $T$, one of these triangles, is colored with other triangles of $S_{12}$. Without loss of generality, we can assume that the edges of this triangle of $S_{12}$ are $a_1, a_2, a_3$. Thus the set of edges leaving $T$, are colored with $a$ and $a_4$. Two of them are colored with $a_4$, and one is colored with $a$.

Let $M$ be a perfect matching of $S_{12}$ containing the edges $a$ and $a_3$. By (c) of Lemma \ref{lem:propHcolorings}, we have that $F=f^{-1}(M)$ is a perfect matching in $H$. Now, observe that $|F\cap \partial_H(V(T))|=1$. Consider the cubic graph $H''$ obtained from $H$ by contracting $T$. Observe that $F\backslash (F \cap  E(T))$ is a perfect matching of $H''$. This violates the definition of the set $U$, since we found a smaller subset of vertices, whose replacement with triangles was leading to a cubic graph containing a perfect matching.

Thus, all triangles of $H$ corresponding to $U$ are colored with the edges $x, y, z$ of the central triangle of $S_{12}$.

Observe that $G$ can be obtained from $H$ by contracting all the triangles corresponding to $U$. Now, using the $S_{12}$-coloring of $H$, we obtain an $S_{10}$-coloring of $G$. Contract all triangles of $H$ corresponding to $U$ and the central triangle of $S_{12}$ to obtain $S_{10}$, and re-color the edges of $H$ having color $x$ with the color $a$, the edges of $H$ with color $y$ with color $b$ and finally, the edges of $H$ with color $z$ with color $c$, respectively. Since $x,y,z$ form an even subgraph in $S_{12}$, by (f) of Lemma \ref{lem:propHcolorings}, the edges of $f^{-1}(\{x,y,z\})$ will form an even subgraph, that is vertex-disjoint union of cycles. Hence, after the re-coloring we obtain an $S_{10}$-coloring of $G$. The proof of the theorem is complete.
\end{proof}

\begin{theorem} Conjectures \ref{conj:5CDC} and \ref{conj:P12conjecture} imply Conjecture \ref{conj:P10conjecture}.
\end{theorem}

\begin{proof} Assume that Conjectures \ref{conj:5CDC} and \ref{conj:P12conjecture} are true, and let $G$ be a bridgeless cubic graph. Let us show that $G$ admits a $P_{10}$-coloring. If $k(G)\leq 4$, then by Conjecture \ref{conj:P12conjecture} it has a $P_{12}$-coloring. Since $P_{12}$ admits a $P_{10}$-coloring, we have that $G$ admits a $P_{10}$-coloring. Thus, without loss of generality, we can assume that $k(G)\geq 5$.

Consider the graph $H$ obtained from $G$ by replacing all vertices of $G$ with triangles. Observe that $H$ is a claw-free bridgeless cubic graph. Hence by Conjectures \ref{conj:5CDC} and \ref{conj:ClawFreeBerge}, and Theorem \ref{thm:Giuseppe5CDCClawFreeEq}, $k(H)\leq 4$. Thus, by Conjecture \ref{conj:P12conjecture}, $H$ admits a $P_{12}$-coloring. Since $P_{12}$ admits a $P_{10}$-coloring, we have that $H$ admits a $P_{10}$-coloring $f$. Observe that since $P_{10}$ is triangle-free, we have that for any triangle $T$ of $H$ there is a vertex $v$ of $P_{10}$, such that $f(E(T))=\partial_{P_{10}}(v)$. Thus, if we contract all the triangles of $H$ that correspond to vertices of $G$, we will obtain a $P_{10}$-coloring of $G$. The proof of the theorem is complete.
\end{proof}

\section{Future work}
\label{sec:future}

In this section, we discuss some open problems and conjectures that are interesting in our point of view. In the previous section, we established a connection between Conjectures \ref{conj:P12conjecture} and \ref{conj:P10conjecture}, and Conjectures \ref{conj:S12conjecture} and \ref{conj:S10Conjecture}. We suspect that this relationship can be extended to a linear order among these four conjectures. Related with this, we would like to offer:

\begin{conjecture}
\label{conj:P10S12implicationConj} Conjecture \ref{conj:P10conjecture} implies Conjecture \ref{conj:S12conjecture}.
\end{conjecture}

All hereditary classes that we discussed up to now either had or are conjectured to have a basis with one element. One may wonder whether there is a hereditary class of cubic graphs arising from an interesting graph theoretic property, such that the basis of the class contains at least two graphs. For a positive integer $k$ let $\mathcal{C}_k$ be the class of connected cubic graphs $G$ with $\chi'_N(G)\leq k$. (i) of Lemma \ref{lem:propHcolorings} implies that $\mathcal{C}_k$ is a hereditary class of cubic graphs. Recently, it was shown that for any simple cubic graph $\chi'_N(G)\leq 7$ \cite{NormalCubics7}. By using this result, a simple inductive proof can be obtained for the following extension of this result:

\begin{theorem}
\label{thm:Multigraphs7} Let $G$ be a cubic graph admitting a normal $k$-edge-coloring for some integer $k$. Then $\chi'_N(G)\leq 7$.
\end{theorem} 

Theorem \ref{thm:Multigraphs7} suggests that $\mathcal{C}_k$ is meaningful when $k=3,4,5,6,7$. Below, we discuss these classes for each of these values. When $k=3$, $\mathcal{C}_k$ represents the class of connected 3-edge-colorable cubic graphs. Thus, our notation is consistent with that of Section \ref{sec:hereditary}. When $k=4$, it can be easily seen that a cubic graph admits a normal 4-edge-coloring, if and only if it admits a $3$-edge-coloring. Thus, $\mathcal{C}_4=\mathcal{C}_3$. When $k=5$, Jaeger has shown that a cubic graph admits a $P_{10}$-coloring if and only if it admits a normal 5-edge-coloring. On the other hand, we have that any cubic graph admitting a $P_{10}$-coloring, has to be bridgeless. Thus, if Conjecture \ref{conj:P10conjecture} is true, then $\mathcal{C}_5=\mathcal{C}_b$. Finally, when $k=6$ or $k=7$, we suspect that the bases of the classes $\mathcal{C}_6$ and $\mathcal{C}_7$ contain infinitely many cubic graphs. We are able to show that the basis of $\mathcal{C}_7$ must contain at least two graphs. Let $\mathcal{B}$ be any basis of $\mathcal{C}_7$. It can be easily seen that we can assume that it does not contain a 3-edge-colorable graph. Moreover, by a simple inductive proof, one can show that all elements of $\mathcal{B}$ can be assumed to be simple graphs. Now, let $S_{16}$ be the graph from Figure \ref{fig:S16}. The following two results are proved in \cite{AnushSylvester}:

\begin{figure}[ht]
  \begin{center}
	
	\tikzstyle{every node}=[circle, draw, fill=black!50,
                        inner sep=0pt, minimum width=4pt]
	
		\begin{tikzpicture}
																							
			\node[circle,fill=black,draw] at (-5.5,-1) (n1) {};
			
			\node[circle,fill=black,draw] at (-6, -0.5) (n2) {};
			\node[circle,fill=black,draw] at (-6, 0.5) (n22) {};
																								
			\node[circle,fill=black,draw] at (-5,-0.5) (n3) {};
			\node[circle,fill=black,draw] at (-5,0.5) (n33) {};
																								
			\node[circle,fill=black,draw] at (-3.5,-1) (n4) {};

			\node[circle,fill=black,draw] at (-4, -0.5) (n5) {};
			\node[circle,fill=black,draw] at (-4, 0.5) (n55) {};
																								
			\node[circle,fill=black,draw] at (-3,-0.5) (n6) {};
			\node[circle,fill=black,draw] at (-3, 0.5) (n66) {};
																								
			\node[circle,fill=black,draw] at (-1.5,-1) (n7) {};
																								
			\node[circle,fill=black,draw] at (-2, -0.5) (n8) {};
			\node[circle,fill=black,draw] at (-2, 0.5) (n88) {};
																								
			\node[circle,fill=black,draw] at (-1,-0.5) (n9) {};
			\node[circle,fill=black,draw] at (-1, 0.5) (n99) {};
																								
			\node[circle,fill=black,draw] at (-3.5,-2) (n10) {};

			\path[every node]
			(n1) edge  (n2)

			edge  (n3)
			edge (n10)
																								   	
			(n2) edge (n22)
			    edge (n33)
																								       
			(n3) edge (n22)
			    edge (n33)
			(n4) edge (n5)
			edge (n6)
			edge (n10)
			
			(n22) edge (n33)
																								    
			(n5) edge (n55)
			    edge (n66)
			(n6) edge (n55)
			    edge (n66)
			    
			 (n55) edge (n66)
																								   
			(n7) edge (n8)
			edge (n9)
			edge (n10)
																								    
			(n8) edge (n88)
			    edge (n99)
		(n9) edge (n88)
		    edge (n99)
		   
		 (n88) edge (n99)
																								  
			;
		\end{tikzpicture}
																
	\end{center}
								
	\caption{The graph $S_{16}$.}
	\label{fig:S16}

\end{figure}
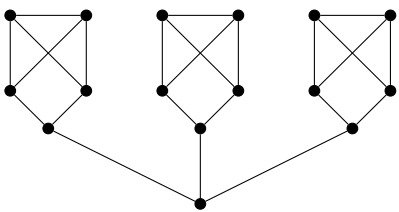


\begin{theorem}
\label{thm:S16Simple} Let $G$ be a simple graph with $G\prec S_{16}$. Then $G= S_{16}$. 
\end{theorem}

\begin{theorem}
\label{thm:PetersenSimple} Let $G$ be a simple graph with $G\prec P_{10}$. Then $G= P_{10}$. 
\end{theorem}

Theorems \ref{thm:S16Simple} and \ref{thm:PetersenSimple} suggest that the only way to color the graphs $S_{16}$ and $P_{10}$ with simple graphs is to take them in the basis $\mathcal{B}$. Thus, $\mathcal{B}$ must contain at least two graphs. 

\medskip

Finally, we would like to discuss some algorithmic problems. For a fixed connected cubic graph $H$ consider a decision problem which we call the $H$-problem:

\begin{problem}
\label{prob:Hproblem} ($H$-problem) Given a connected cubic graph $G$, the goal is to decide whether $G$ admits an $H$-coloring.
\end{problem}

Observe that when $H$ is $3$-edge-colorable, we have that $H$-problem is equivalent to testing $3$-edge-colorability of the input graph $G$, which is NP-complete \cite{Holyer}. When $H=S_{10}$, we have that all instances of $H$-problem have a trivial ``yes" answer provided that Conjecture \ref{conj:S10Conjecture} is true. Thus, this problem is polynomial time solvable if Conjecture \ref{conj:S10Conjecture} is true. When $H=S_{12}$, Conjecture \ref{conj:S12conjecture} implies that $H$-problem is equivalent to testing the existence of a perfect matching in the input graph $G$. This is known to be polynomial-time solvable. When $H=P_{10}$, Conjecture \ref{conj:P10conjecture} implies that $H$-problem is equivalent to testing bridgelessness of the input graph $G$. This problem is also polynomial time solvable. Finally, when $H=P_{12}$, Conjecture \ref{conj:P12conjecture} implies that $H$-problem is equivalent to testing whether the input graph $G$ can be covered with four perfect matchings. The latter problem is proved to be NP-complete in \cite{EsperetGiuseppe}. Thus, depending on the choice of $H$, the $H$-problem may or may not be NP-complete. Let $\mathcal{C}_{NP}$ be the class of all connected cubic graphs $H$, for which the $H$-problem is NP-complete. We suspect that:

\begin{conjecture}
\label{conj:NPhereditaryClass} $\mathcal{C}_{NP}$ is a hereditary class of cubic graphs.
\end{conjecture}

We also would like to offer the following conjecture, which presents a dichotomy for $H$-problems:

\begin{conjecture}
\label{conj:dichotomy} Let $H$ be a connected cubic graph. Then:
\begin{itemize}
    \item if $H$ admits a $P_{12}$-coloring, then the $H$-problem is NP-complete;
    \item if $H$ does not admit a $P_{12}$-coloring, then the $H$-problem is polynomial-time solvable.
\end{itemize}
\end{conjecture}

\section*{Acknowledgement}

We thank Giuseppe Mazzuoccolo for proving Theorem \ref{thm:Giuseppe5CDCClawFreeEq}.





\begin{thebibliography}{99}

\bibitem{Celmins1984} A. U. Celmins, On cubic graphs that do not have an edge-$3$-colouring, Ph.D. Thesis, Department of Combinatorics and Optimization, University of Waterloo, Waterloo, Canada, 1984.

\bibitem{ChudSeyClawFreeChar} M. Chudnovsky, P. Seymour. The structure of claw-free graphs. In Surveys in combinatorics 2005, London Math. Soc. Lecture Note Ser. 327, pages 153--171. Cambridge Univ. Press, Cambridge, 2005.

\bibitem{EsperetGiuseppe} L. Esperet, G. Mazzuoccolo, On cubic bridgeless graphs whose edge-set cannot be covered by four perfect matchings, J. Graph Theory 77(2), (2014), pp. 144--157.


\bibitem{Fulkerson} D.R. Fulkerson, Blocking and anti-blocking pairs of polyhedra, Math. Programming
1 (1971), 168--194.

\bibitem{AnushSylvester} A. Hakobyan, V. Mkrtchyan, On Sylvester Colorings of Cubic Graphs, submitted (available at https://arxiv.org/abs/1511.02475).


\bibitem{Holyer} I. Holyer, The NP-completeness of edge-coloring, SIAM J. Comput. 10(4), 1981, pp. 718--720.

\bibitem{CQJoins} X. Hou, H.-J. Lai, C.-Q. Zhang, On perfect matching coverings and even subgraph coverings, J. Graph Theory 81(1), 2016, pp. 83--91.

\bibitem{Jaeger1985} F. Jaeger, On five-edge-colorings of cubic graphs and nowhere-zero flow problems, Ars Combinatoria, 20-B, (1985), 229--244.

\bibitem{Jaeger} F. Jaeger, Nowhere-zero flow problems, Selected topics in graph theory, 3, Academic
Press, San Diego, CA, 1988, pp. 71--95.


\bibitem{NormalCubics7} G. Mazzuoccolo, V. Mkrtchyan, Normal edge-colorings of cubic graphs, submitted (availabe at https://arxiv.org/abs/1804.09449).

\bibitem{PetersenRemark}  V. V. Mkrtchyan, A remark on Petersen coloring conjecture of Jaeger, Australasian Journal of Combinatorics 56, (2013), 145--151.

\bibitem{sang-il_oum:2011}
S.-il~Oum, Perfect matchings in claw-free cubic graphs, The Electronic Journal of Combinatorics 18(1), 2011.

\bibitem{Preiss1981} M. Preissmann, Sur les colorations des aretes des graphes cubiques, These de $3$-eme cycle, Grenoble (1981).

\bibitem{Samal2017} R. \v{S}\'{a}mal, Cycle-continuous mappings-order structure, J. Graph theory 85(1), 2017, pp. 56--73.

\bibitem{Schrijver} A. Schrijver, Combinatorial Optimization, Springer, New York, 2003.


\bibitem{Seymour} P. D. Seymour, On multicolourings of cubic graphs, and conjectures of Fulkerson and Tutte. Proc. London Math. Soc. 38 (3), 423--460, 1979.



\bibitem{SteffenP} E. Steffen, personal communication, 2011.



\bibitem{West} D. B. West, Introduction to Graph Theory, Prentice-Hall,
Englewood Cliffs, 1996.

\bibitem{Zhang1997} C.-Q. Zhang, Integer flows and cycle covers of graphs, Marcel Dekker, Inc., New York Basel Hong Kong, 1997.



\end{thebibliography}
\end{document}